\documentclass[conference,a4paper]{IEEEtran}
\usepackage{mleftright}       
\mleftright                   
\usepackage[utf8]{inputenc}
\usepackage[T1]{fontenc}
\usepackage[cmex10]{amsmath}
\usepackage{amssymb}
\usepackage{dsfont}
\usepackage{xcolor}
\usepackage{tabu}
\usepackage{adjustbox}
\usepackage{multirow}
\usepackage{amsthm}
\usepackage{adjustbox}
\usepackage{floatrow}
\usepackage{verbatim}
\usepackage{nicematrix}
\usepackage{tablefootnote}
\usepackage[maxbibnames=99,style=ieee,sorting=nyt,citestyle=numeric-comp]{biblatex}
\newcommand{\Mod}[1]{\ (\mathrm{mod}\ #1)}

\usepackage{scrextend}

\makeatletter
\let\MYcaption\@makecaption
\makeatother

\usepackage[font=footnotesize]{subcaption}

\makeatletter
\let\@makecaption\MYcaption
\makeatother

\bibliography{bibliography.bib}

\title{Nearest Neighbor Representations of Neural Circuits}
\author{%
   \IEEEauthorblockN{\textbf{Kordag Mehmet Kilic}\IEEEauthorrefmark{1}, \textbf{Jin Sima}\IEEEauthorrefmark{2} and \textbf{Jehoshua Bruck}\IEEEauthorrefmark{1}}
   \IEEEauthorblockA{\IEEEauthorrefmark{1}%
   Electrical Engineering, California Institute of Technology, USA, \texttt{\{kkilic,bruck\}@caltech.edu}
   }
   \IEEEauthorblockA{\IEEEauthorrefmark{2}%
   Electrical and Computer Engineering, University of Illinois Urbana-Champaign, USA, \texttt{jsima@illinois.edu}
   }
 }

\newtheorem{theorem}{Theorem}
\newtheorem{corollary}{Corollary}[theorem]

\newtheorem{definition}{Definition}
\newtheorem{conjecture}{Conjecture}

\DeclareMathOperator {\diag}{diag}

\usepackage{tikz}[border=2mm]
\usepackage{tikz-3dplot}
\usepackage{float}
\usetikzlibrary{shapes.geometric,positioning,shapes.gates.logic.US,tikzmark,calc,ext.transformations.mirror}
\AtBeginBibliography{\small}
\tikzset{
    ->,
    gate/.style={draw=black,fill=#1,minimum width=5mm,circle},
    square/.style={minimum width=6mm,regular polygon,regular polygon sides=4},
    triangle/.style={minimum width=4mm,regular polygon, regular polygon sides=3},
    every pin edge/.style={draw=black},
    GateCfg/.style={
            logic gate inputs={normal,normal,normal},
            draw,
            scale=2
        }
}

\begin{document}

\maketitle
\thispagestyle{plain}
\pagestyle{plain}

\begin{abstract}
    Neural networks successfully capture the computational power of the human brain for many tasks. Similarly inspired by the brain architecture, Nearest Neighbor (NN) representations is a novel approach of computation. We establish a firmer correspondence between NN representations and neural networks. Although it was known how to represent a single neuron using NN representations, there were no results even for small depth neural networks. Specifically, for depth-2 threshold circuits, we provide explicit constructions for their NN representation with an explicit bound on the number of bits to represent it. Example functions include NN representations of convex polytopes (AND of threshold gates), IP2, OR of threshold gates, and linear or exact decision lists.
\end{abstract}

\section{Introduction}
\label{sec:intro}

The study of Nearest Neighbors (NN) started in 1950s for regression and classification and they remained one of the most fundamental models in machine learning \cite{fix1951discriminatory}. Intuitively, each concept is thought to be a real-valued vector in real space and the closer the concepts are, the more ``similar'' they should be. Therefore, information theoretic aspects of NN models were intriguing and in this regard, the minimum size and description error trade-off was questioned, which became a practical issue for NN models called the infamous \textit{curse of dimensionality} \cite{cover1967nearest,weber1998quantitative,beyer1999nearest}. In addition, NN models gained more attention recently since Natural Language Processing (NLP) pipelines may contain large vector databases to store high dimensional real-valued \textit{embeddings} of words, sentences, or even images \cite{johnson2019billion,lewis2020retrieval}. Thus, the theoretical limits of the database sizes, the approximations of the NN search algorithms, and the design of ``meaningful'' embeddings became important research directions \cite{indyk1998approximate,kushilevitz1998efficient,malkov2018efficient,reimers2019sentence}. An example use of vector databases in NLP involving NN search algorithms is given in Figure \ref{fig:rag}.

Alternatively, NNs could play a role in our understanding on how the brain works. Inspired by the architecture of the brain, \textit{neural networks} became the backbone of many machine learning models, and similarly, \textit{Nearest Neighbor (NN) Representations} are introduced as emerging models of computation \cite{minsky2017perceptrons,hajnal2022nearest,kilic2023information}. NN representations mimic the integrated structure of memory and computation in the brain where the concepts in the memory are embedded in \textit{anchors}. They are specifically defined for Boolean functions where the input vectors are binary vectors, i.e., $X \in \{0,1\}^n$, and the labels can be either \textbf{\textcolor{red}{0}} (\textbf{\textcolor{red}{red}}) or \textcolor{blue}{1} (\textcolor{blue}{blue}). We always use a single NN to find the correct label, i.e., the closest anchor gives the output label of the input vector $X$. Two examples of NN representations for $2$-input Boolean functions AND and OR are given in Figure \ref{fig:and-or}.

\begin{figure}[h!]
    \centering
    \includegraphics[width=\linewidth]{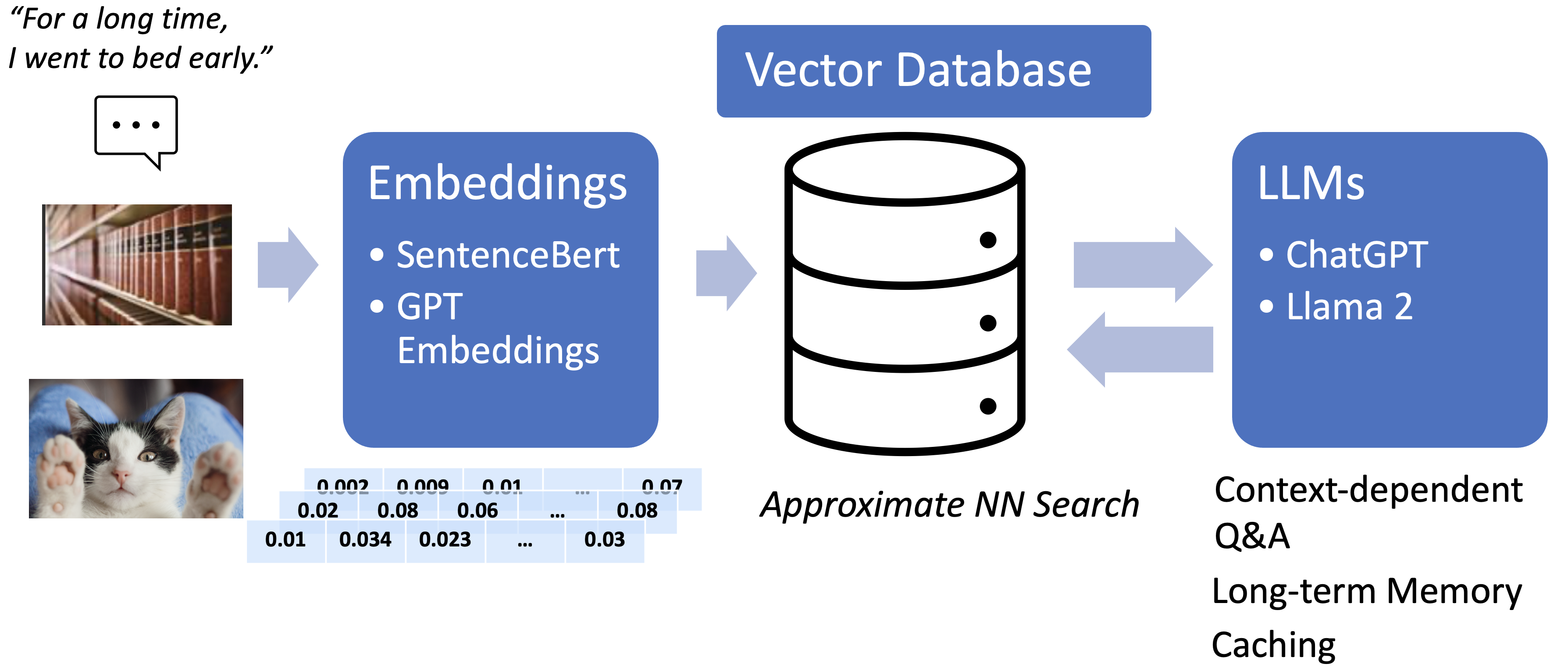}
    \caption{An illustration of an \textit{Retrieval Augmented Generation (RAG)} pipeline in NLP \cite{lewis2020retrieval}. Each ``document'', which can be a single sentence, a book, or an image of a cat, is converted into a high-dimensional embedding by \textit{Embedding Models} and is stored in a vector database. To do knowledge-intensive tasks, similar ``documents'' can be found by approximate NN search algorithms.}
    \label{fig:rag}
\end{figure}




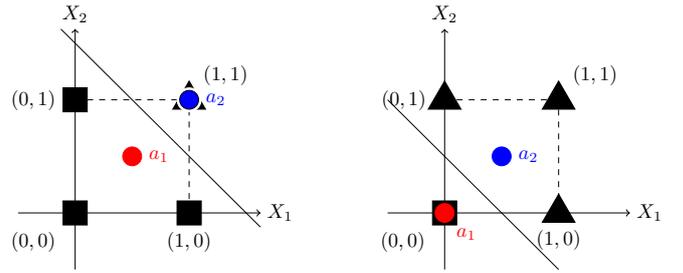
\begin{figure}[h!]
    \centering
    \begin{floatrow}
    \ffigbox[\FBwidth]
    {
    \scalebox{0.75}{
        \begin{tikzpicture}
        
        \draw[->] (-1,0) -- (3.25,0) node[right] {$X_1$}; 
        \draw[->] (0,-1) -- (0,3.25) node[above] {$X_2$};
        
        \newcommand{\comma}{,}
        \coordinate (c00) at (0,0);
        \coordinate (c01) at (0,2);
        \coordinate (c10) at (2,0);
        \coordinate (c11) at (2,2);
        
        \node[gate,square,color=black,label=225:$(0\comma0)$] (a00) at (0,0) {};
        \node[gate,square,color=black,label=180:$(0\comma1)$] (a01) at (0,2) {};
        \node[gate,square,color=black,,label=270:$(1\comma0)$] (a10) at (2,0) {};
        \node[gate,triangle,color=black,,label=45:$(1\comma1)$] (a11) at (2,2) {};
        
        \draw[dashed] (a01|-a11) -- (a11) -- (a11|-a10);
        
        \node[gate,minimum width=3mm,color=red,label=0:$\textcolor{red}{a_1}$] (anchor-1) at (1,1) {};
        
        \node[gate,minimum width=4.2mm,fill=white,draw=white] (anchor-2) at (2,2) {};
        \node[gate,minimum width=3mm,fill=blue,label=0:$\textcolor{blue}{a_2}$] (anchor-2) at (2,2) {};
        
        \coordinate (cy) at (-0.25,3.25);
        \coordinate (cx) at (3.25,-0.25);
        \draw[-] (cy) -- (cx);
        
        \end{tikzpicture}
}
\label{fig:and_coor}

}

{
\scalebox{0.75}{
        \begin{tikzpicture}
            
            \draw[->] (-1,0) -- (3.25,0) node[right] {$X_1$}; 
            \draw[->] (0,-1) -- (0,3.25) node[above] {$X_2$};
            
            \newcommand{\comma}{,}
            \coordinate (c00) at (0,0);
            \coordinate (c01) at (0,2);
            \coordinate (c10) at (2,0);
            \coordinate (c11) at (2,2);
            
            \node[gate,square,color=black,label=225:$(0\comma0)$] (a00) at (0,0) {};
            \node[gate,triangle,color=black,label=180:$(0\comma1)$] (a01) at (0,2) {};
            \node[gate,triangle,color=black,,label=270:$(1\comma0)$] (a10) at (2,0) {};
            \node[gate,triangle,color=black,,label=45:$(1\comma1)$] (a11) at (2,2) {};
            
            \draw[dashed] (a01|-a11) -- (a11) -- (a11|-a10);
            
            \node[gate,minimum width=0.5mm,color=red,label=300:$\textcolor{red}{a_1}$] (anchor-1) at (0,0) {};
            
            \node[gate,minimum width=0.5mm,color=blue,label=0:$\textcolor{blue}{a_2}$] (anchor-2) at (1,1) {};
            
            \coordinate (cy) at (-1,2);
            \coordinate (cx) at (2,-1);
            \draw[-] (cy) -- (cx);
            
        \end{tikzpicture}
    }
}
\end{floatrow}
    \caption{NN representations for $2$-input Boolean functions $\text{AND}(X_1,X_2)$ (left) and $\text{OR}(X_1,X_2)$ (right). Triangles denote $f(X) = 1$ and squares denote $f(X) = 0$. It can be seen that \textbf{\textcolor{red}{red}} anchors are closest to squares and  \textcolor{blue}{blue} anchors are closest to triangles. A separating line between anchors pairs is drawn.}
    \label{fig:and-or}
\end{figure}

Given that neural networks and NN representations draw inspiration from the architecture of the brain, how do they relate? Could we find an NN representation for a neural network? A more ambitious pursuit would involve exploring whether we can substitute a Large Language Model (LLM) with an equivalent NN representation. In this work, we start answering this question by finding NN representations for a family of depth-2 neural networks, which was not done before. We now introduce the formal definitions for NN representations. Let $d(a,b)$ be the Euclidean distance between two real vectors $a$ and $b$. We note that we are naturally interested in smallest size NN representations for a given Boolean function, which is called the \textit{NN complexity}.

\begin{definition}
    The \textup{Nearest Neighbor (NN) Representation} of a Boolean function $f$ is a set of anchors consisting of the disjoint subsets $(P,N)$ of $\mathbb{R}^n$ such that for every $X \in \{0,1\}^n$ with $f(X) = 1$, there exists $p \in P$ such that for every anchor $n \in N$, $d(X,p) < d(X,n)$, and vice versa. The \textup{size} of the NN representation is $|P \cup N|$.
\end{definition}

\begin{definition}
    The \textup{Nearest Neighbor Complexity} of a Boolean function $f$ is the minimum size over all NN representations of $f$, denoted by $NN(f)$.
\end{definition}
To quantify the information capacity of an NN representation, we consider the amount of bits to represent an anchor in real space. Without loss of generality, one can assume that the entries can be rational numbers when the inputs are discrete. We define \textit{anchor matrix} $A \in \mathbb{Q}^{m\times n}$ of an NN representation where each row is an $n$-dimensional anchor and the size of the representation is $m$. The \textit{resolution of an NN representation} is the maximum number of bits required to represent an entry of the anchor matrix.

\begin{definition}
    The \textup{resolution ($RES$) of a rational number} $a/b$ is $RES(a/b) = \lceil \max\{\log_2{|a+1|},\log_2{|b+1|}\}\rceil$ where $a,b \in \mathbb{Z}$, $b \neq 0$, and they are coprime.
    
    For a matrix $A \in \mathbb{Q}^{m\times n}$, $RES(A) = \max_{i,j} RES(a_{ij})$. The \textup{resolution of an NN representation} is $RES(A)$ where $A$ is the corresponding anchor matrix.
\end{definition}

In Figure \ref{fig:and-or}, the $2$-input AND function has two anchors $\textcolor{red}{a_1} = (0.5,0.5)$ and $\textcolor{blue}{a_2} = (1,1)$. By using the definition, we see that the resolution of this representation is $\lceil \log_2{3} \rceil = 2$.

To find a NN representation for a depth-2 neural network, the first step is to understand the NN representations of a single neuron. Historically, each neuron in a neural network was a \textit{linear threshold function}. A linear threshold function is a weighted linear summation fed to a step function. Equivalently, for a linear threshold function $f(X)$, we have $f(X) = \mathds{1}\{w^T X \geq b\}$ where $\mathds{1}\{.\}$ is an indicator function with outputs $\{0,1\}$, $w \in \mathbb{Z}^n$ is an integer weight vector, $X \in \{0,1\}^n$ is the binary input vector, and $b \in \mathbb{Z}$ is a bias term. Similarly, we define an \textit{exact threshold function} using an equality check, i.e., $f(X) = \mathds{1}\{w^T X = b\}$. Together, we refer to both of them as \textit{threshold functions} and a device computing a threshold function is called a \textit{threshold gate}. 

It is already known that a Boolean function with $NN(f) = 2$ must be a linear threshold function with resolution $O(n\log{n})$ \cite{hajnal2022nearest,kilic2023information}. Remarkably, neurons in a neural network have the smallest NN complexity. NN representation size and resolution trade-offs for threshold functions are treated in \cite{kilic2024nearest}.

Let $\text{LT}$ denote the set of linear threshold functions and $\text{ELT}$ denote the set of exact threshold functions. We use $\circ$ to compose threshold gates into circuits, e.g., $\text{ELT}\circ\text{LT}$ is a depth-2 threshold circuit where an exact threshold gate is on the top (output) layer and the bottom layer has linear threshold gates. In Figure \ref{fig:ldl_dom}, a $5$-input threshold circuit in $\text{LT}\circ\text{LT}$ is given.

In a neural network, what does happen if we use devices with higher NN complexity in each building block? We delve into this question by considering another family of Boolean functions, called \textit{symmetric Boolean functions}. Symmetric Boolean functions have the same output value among $X$s with the same number of $1$s. We use SYM to denote the set of symmetric Boolean functions. These functions are significant because they appear in other complexity theory results and an important property is that any $n$-input Boolean function can be interpreted as a $2^n$-input symmetric Boolean function \cite{bruck1990harmonic,haastad1995top,siu1991depth,stockmeyer1976combinational}.

The NN complexity of symmetric Boolean functions is treated in previous works \cite{hajnal2022nearest,kilic2023information}. An important result is that the NN complexity of a symmetric Boolean function is at most the number of \textit{intervals} it contains \cite{kilic2023information}. Let $|X|$ denote the number of $1$s of a binary vector $X \in \{0,1\}^n$. 

\begin{definition}
    An $n$-input Boolean function is called \textup{symmetric} if $f(X)$ has the same value for all permutations of the input $X$. Namely, a symmetric Boolean function $f(X)$ is a function of $|X|$.

    An \textup{interval} of a symmetric Boolean function is the interval $[a,b]$ where $f(X)$ is constant for $|X| \in [a,b]$, $a \geq 0$, $b \leq n$, and decreasing $a$ or increasing $b$ is not possible. The $\text{j}^{th}$ interval is denoted by $[I_{j-1}+1,I_{j}]$ where $I_0 = -1$. The total number of intervals of $f(X)$ is denoted by $I(f)$.
\end{definition}

For a $5$-input symmetric Boolean function $f(X)$ with $I(f) = 4$ intervals, the interval boundaries are illustrated in Eq. \eqref{eq:sym_ex}. Intervals of a symmetric Boolean function is significant in other Circuit Complexity Theory results as well \cite{muroga1959logical,minnick1961linear,siu1991depth}.

\begin{equation}
    \label{eq:sym_ex}
    \begin{tabular}{c|c c c}
        $|X|$ & $f(X)$ \\
        \cline{1-2}
         0 & \textcolor{red}{\textbf{0}} & \hphantom{a} & $I_{1} = 0$\\
         1 & \textcolor{blue}{1} & \hphantom{a} & \\ 
         2 & \textcolor{blue}{1} & \hphantom{a} & $I_{2} = 2$ \\
         3 & \textcolor{red}{\textbf{0}} & \hphantom{a} & $I_3 = 3$\\
         4 & \textcolor{blue}{1} & \hphantom{a} & \\
         5 &  \textcolor{blue}{1}& \hphantom{a} & $I_{4} = 5$\\
    \end{tabular}
\end{equation}

By composing threshold functions and symmetric Boolean functions together, one can obtain different families of Boolean functions. For instance, a \textit{linear decision list} (denoted by LDL) of depth $m$ is a sequential list of linear threshold functions $f_1(X),\dots,f_m(X)$ where the output is $z_k \in \{0,1\}$ for $f_i(X) = 0$ for $i < k$ and $f_k(X) = 1$ and it is $z_{m+1} \in \{0,1\}$ if all $f_i(X) = 0$. An LDL can be computed in a depth-2 linear threshold circuit, i.e., $\text{LDL} \subseteq \text{LT}\circ\text{LT}$ by Eq. \eqref{eq:ldl_lt} \cite{turan1997linear}. One can similarly define the class of EDL where exact threshold functions are used and $\text{EDL} \subseteq \text{LT}\circ\text{ELT}$. We give an example of a $5$-input LDL of depth $3$ in Figure \ref{fig:ldl}.

\begin{equation}
    \label{eq:ldl_lt}
    l(X) = \mathds{1}\Big\{ \sum_{i=1}^m (-1)^{z_i-1} 2^{m-i}f_i(X) \geq 1-z_{m+1} \Big\}
\end{equation}

\begin{figure}[h!]
    \centering
    \begin{tikzpicture}
    \newdimen\nodeDist
    \nodeDist=20mm
    \node [rectangle,draw] (A) {\footnotesize$\mathds{1}\{x_1 + x_2 \geq 1\}$};
    \path (A) ++(-120:\nodeDist) node [rectangle,draw] (B) {\footnotesize$\mathds{1}\{2x_1 + x_3 + x_4 \geq 2\}$};
    \path (A) ++(-60:\nodeDist) node (C) {\textcolor{blue}{1}};
    \path (B) ++(-120:\nodeDist) node [rectangle,draw] (D) {\footnotesize$\mathds{1}\{x_2-x_5 \geq 0\}$};
    \path (B) ++(-60:\nodeDist) node (E) {\textcolor{red}{\textbf{0}}};
    \path (D) ++(-120:\nodeDist) node (F) {\textcolor{blue}{1}};
    \path (D) ++(-60:\nodeDist) node (G) {\textcolor{red}{\textbf{0}}};

    \draw (A) -- (B) node [left,pos=0.25] {0};
    \draw (A) -- (C) node [right,pos=0.25] {1};
    \draw (B) -- (D) node [left,pos=0.25] {0};
    \draw (B) -- (E) node [right,pos=0.25] {1};
    \draw (D) -- (F) node [left,pos=0.25] {0};
    \draw (D) -- (G) node [right,pos=0.25] {1};
\end{tikzpicture}
    \caption{A Linear Decision List of Depth 3 with $5$ binary inputs.}
    \label{fig:ldl}
\end{figure}
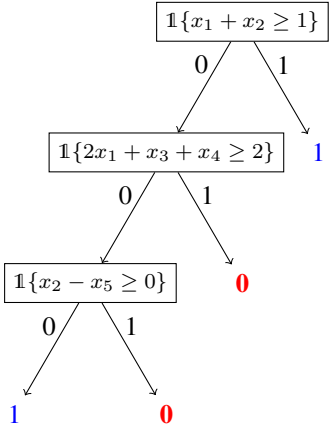

In Eq. \eqref{eq:ldl_lt}, the implied top gate (see Figure \ref{fig:ldl_dom}) is actually called a DOMINATION gate (denoted by DOM) because the leading one in the vector $(f_1(X),\dots,f_m(X))$ \textit{dominates} the importance of the others. Examples of DOM gates are treated in different contexts \cite{beigel1994perceptrons,kilic2021neural}. The definition in Eq. \eqref{eq:dom} uses an abuse of notation to emphasize that any sign combinations for the weights is allowed. In Figure \ref{fig:ldl_dom}, a depth-2 circuit construction of the linear decision list depicted in Figure \ref{fig:ldl} is given.

\begin{equation}
    \label{eq:dom}
    \text{DOM}(X) = \mathds{1}\Big\{\sum_{i=1}^n \pm 2^{n-i}x_i \geq b\Big\} 
\end{equation}

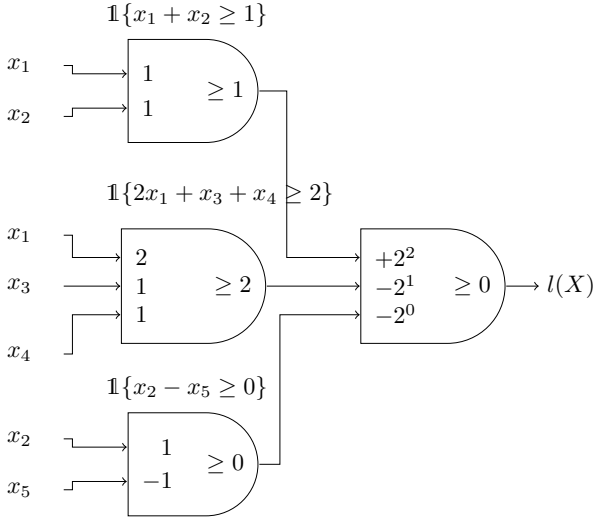
\begin{figure}[h!]
    \centering
    \scalebox{0.9}{\begin{tikzpicture}
    \tikzstyle{sum} = [gate=white,label=center:+]
    \tikzstyle{input} = [circle]
    \newcommand{\eq}{=}

    \node[input,label=180:$x_1$] (i-1) at (0,-0.5) {};
    \node[input,label=180:$x_2$] (i-2) at (0,-1.25) {};
    \node[input,label=0:$\mathds{1}\{x_1+x_2 \geq 1\}$] (EQ-12) at (0.5,0.25) {};
    
    \node[and gate US, draw, logic gate inputs=nn, scale=3] (TH1) at (2,-0.875) {};
        \draw (i-1) -| +(0.3,0) |- (TH1.input 1);
        \draw (i-2) -| +(0.3,0) |- (TH1.input 2);
    
    \node [right=0.1 of TH1.input 1] {$1$};
    \node [right=0.1 of TH1.input 2] {$1$};
    \node [left =0.1 of TH1.output] {$\geq 1$};
    
    \node[input,label=180:$x_1$] (i-4) at (0,-3) {};
    \node[input,label=180:$x_3$] (i-5) at (0,-3.75) {};
    \node[input,label=180:$x_4$] (i-6) at (0,-4.75) {};
    
    \node[input,label=0:$\mathds{1}\{2x_1 + x_3 + x_4 \geq 2\}$] (EQ-12) at (0.5,-2.4) {};
    \node[and gate US, draw, logic gate inputs=nnn, scale=2.5] (TH2) at (2,-3.75) {};
        \draw (i-4) -| +(0.3,0) |- (TH2.input 1);
        \draw (i-5) -| +(0.3,0) |- (TH2.input 2);
        \draw (i-6) -| +(0.3,0) |- (TH2.input 3);
        
    \node [right=0.1 of TH2.input 1] {$2$};
    \node [right=0.1 of TH2.input 2] {$1$};
    \node [right=0.1 of TH2.input 3] {$1$};
    \node [left =0.1 of TH2.output] {$\geq 2$};

    \node[input,label=180:$x_2$] (i-7) at (0,-6) {};
    \node[input,label=180:$x_5$] (i-8) at (0,-6.75) {};

    \node[input,label=0:$\mathds{1}\{x_2 - x_5 \geq 0\}$] (EQ-1m) at (0.5,-5.25) {};
    \node[and gate US, draw, logic gate inputs=nn, scale=3] (TH3) at (2,-6.375) {};
        \draw (i-7) -| +(0.3,0) |- (TH3.input 1);
        \draw (i-8) -| +(0.3,0) |- (TH3.input 2);
    \node [right=0.1 of TH3.input 1] {$\hphantom{-}1$};
    \node [right=0.1 of TH3.input 2] {$-1$};
    \node [left =0.1 of TH3.output] {$\geq 0$};
    
    \newcommand{\offset}{2.5}
    \node[and gate US, draw, logic gate inputs=nnn, scale=2.5,pin=0:$l(X)$] (DOM) at (\offset+3,-3.75) {};
        \draw (TH1.output) -| +(0.4,0) |- (DOM.input 1);
        \draw (TH2.output) -| +(0.3,0) |- (DOM.input 2);
        \draw (TH3.output) -| +(0.3,0) |- (DOM.input 3);

    \node [right=0.1 of DOM.input 1] {$+2^2$};
    \node [right=0.1 of DOM.input 2] {$-2^1$};
    \node [right=0.1 of DOM.input 3] {$-2^0$};
    \node [left =0.1 of DOM.output] {$\geq 0$};
    \end{tikzpicture}
    }
    \caption{The depth-2 threshold circuit construction of a LDL $l(X)$ with a DOM gate on top. This shows that $l(X) \in \text{LT}\circ\text{LT}$. The signs of the powers of two depends on the labeling of the outputs $z_i$s. If the first layer consists of exact threshold gates, then we have an EDL.}
    \label{fig:ldl_dom}
\end{figure}

More precisely, our main contribution is the explicit constructions of the NN representations of depth-2 circuits under some regularity assumptions where the top gate is a symmetric gate or a DOM gate. 
The circuits under consideration include many important functions in Boolean analysis. We start with the symmetric Boolean function $\text{PARITY}(X) = \bigoplus_{i=1}^n x_i$, which had a lot of attention from the mathematics and information theory for many decades \cite{haastad1986almost,kautz1961realization,paturi1990threshold}. It is known that $\text{PARITY}(X)$ has NN complexity equal to $n+1$ with at most $O(\log{n})$ resolution \cite{hajnal2022nearest,kilic2023information}.


We also analyze the famous INNER-PRODUCT-MOD2 (denoted by IP2) function, $\text{IP2}_{2n}(X,Y) = X^T Y \Mod{2} = \bigoplus_{i=1}^n x_i \land y_i$. The IP2 received treatment in Boolean analysis since it is still an open problem if $\text{IP2} \in \text{LT} \circ \text{LT}$ for polynomially large circuits in $n$ where partial results are known \cite{forster2002linear,hajnal1993threshold}. For the NN representations with unbounded resolution, it is shown that $NN(\text{IP}_{2n}) \geq 2^{n/2}$ \cite{hajnal2022nearest}. Recently, an interesting result for the IP2 is obtained for $k$-NN Representations \cite{dicicco2024nearest}. We give a construction of $\text{IP2}_{2n}$ with $2^n$ many anchors and $O(1)$ resolution using the property $\text{IP2}_{2n} = \text{PARITY} \circ \text{AND}_2 \in \text{SYM} \circ \text{LT}$. This is far from the lower bound but it could be optimal for constant resolution.

Let $\text{EQ}_{2n}$ be the $2n$-input EQUALITY function, which is $\mathds{1}\{X = Y\}$ where $X$ and $Y$ are $n$-bit unsigned integers. Similarly, define $\text{COMP}_{2n}(X,Y) = \mathds{1}\{X \geq Y\}$ for COMPARISON. These functions are important examples of threshold functions since there are many results on their circuit complexity \cite{amano2005complexity,chattopadhyay2018short,roychowdhury1994lower,kilic2021neural,kilic2022algebraic}. We are mainly interested in circuits of these two functions.

The function $\text{OR}_n \circ \text{EQ}_{2n}$ is treated in \cite{chattopadhyay2019lower} to show that for any LDL, the depth of the list must be $2^{\Omega(n)}$ (note that the number of inputs is $2n^2$) using a result by \cite{impagliazzo2010communication}. Another recent result shows that any LDL can be treated as an EDL with a polynomial blow-up in the depth of the list and $\text{LDL} \subsetneq \text{EDL}$ for polynomially large depth in $n$ \cite{dahiya2024linear}. Since $\text{OR}_n \circ \text{EQ}_{2n} \in \text{SYM}\circ\text{ELT}$, we obtain an NN representation with exponentially large number of anchors in $n$ but it is not known if this is tight.

Consider the $m$-input ODD-MAX-BIT function denoted by $\text{OMB}_m$ as the DOM function $\mathds{1}\{\sum_{i=1}^m (-1)^{i-1} 2^{n-i} x_i > 0\}$. A recent result was obtained for $\text{OMB}_m\circ\text{EQ}_{2n} \in \text{DOM}\circ\text{ELT}$, which shows that it is the first explicit example of Boolean functions where any polynomially large size depth-2 linear threshold circuit computing it requires ``large'' weights \cite{chattopadhyay2018short}. The idea is to use the \textit{sign-rank} method from the Communication Complexity Theory. For this function, we give an NN representation with exponentially large number of anchors. 

The organization of the paper is as follows. We first begin examining the NN representations of convex polytopes, namely, $\text{AND}\circ\text{LT}$ circuits. Secondly, we find constructions for the NN representations of $\text{SYM}\circ\text{ELT}$ and $\text{SYM}\circ\text{LT}$. A summary of the results for explicit functions is given in the Table \ref{tab:nn}. Finally, we consider the NN representations of LDLs and EDLs.

\begin{table}[h!]
\floatsetup[table]{style=plaintop}
  \captionsetup[table]{position=top}
\caption{Summary of the Results for NN Representations}
\label{tab:nn}
\centering
\begin{minipage}{\textwidth}
\centering
\begin{adjustbox}{width=.75\textwidth}
\begin{tabular}{|c||c||c|}
    \hline
     \multirow{2}{*}{\textbf{Function}} & \multicolumn{2}{c|}{\textbf{NN Representation}} \\ 
     \cline{2-3} & Size & Resolution \\
     \hline\hline
     $\text{AND}_m \circ \text{EQ}_{2n}$ & $2m + 1$ & $O(n)$ \\
     \hline
     $\text{OR}_m \circ \text{EQ}_{2n}$ & $(m+2)2^{m-1}$ & $O(n)$ \\
     \hline
     $\text{PARITY}_m \circ \text{EQ}_{2n}$ & $3^m$ & $O(n)$ \\
     \hline
     $\text{PARITY}_m \circ \text{COMP}_{2n}$ & $2^m$ & $O(n)$ \\
     \hline
     $\text{IP2}_{2n}$ & $2^n$ &  $O(1)$ \\
     \hline
     $\text{OMB}_m\circ\text{EQ}_{2n}$ & $(m+1)2^m$ & $O(n)$ \\
     \hline
\end{tabular}
\end{adjustbox}
\end{minipage}
\end{table}

\section{NN Representations of Convex Polytopes ($\text{AND}_m \circ \text{LT}_n$)}
\label{sec:polytopes}

The membership function to a convex polytope, i.e., $\mathds{1}\{AX \leq b\}$ can be thought as an $\text{AND} \circ \text{LT}$ circuit where each row is a linear threshold function. We show that $(m+1)$-size NN representation with a single positive anchor is a membership function of a convex polytope with $m$ half-spaces. And, conversely, given a membership function of a convex polytope with $m$ half-spaces, i.e. $\mathds{1}\{AX \leq b\}$, one can construct an $(m+1)$-size NN representation with a single positive anchor.

First direction is easy to see: Given the anchor matrix and anchor labels, we can find optimal separating hyperplanes between them. We provide the proof for the other direction where the geometric idea is given in Figure \ref{fig:polygon}. We use the notation $\diag(AA^T) \in \mathbb{Q}^m$ to denote the squared Euclidean norms of each row of a matrix $A \in \mathbb{Q}^{m\times n}$.


\begin{figure}[h!]
    \centering
    \scalebox{1.2}{
    \begin{tikzpicture}
      \coordinate (A) at (0.75,0.1);
      \coordinate (B) at (2,1);
      \coordinate (C) at (2.4,2);
      \coordinate (D) at (1,2.5);
      \coordinate (E) at (0,1.5);
      \fill[lightgray] (A) -- (B) -- (C) -- (D) -- (E) -- cycle;
      \draw[thick,-] (A) -- (B);
      \draw[thick,-] (B) -- (C);
      \draw[thick,-] (C) -- (D);
      \draw[thick,-] (D) -- (E);
      \draw[thick,-] (E) -- (A);
      \coordinate (a0) at (1.2,1.2);
      \node[above left=0.5 and 0.5 of a0,label=0:$Ax \leq b$] (aa) {};
      \node[right=1.5 of aa,label=0:$Ax > b$] (aa) {};

      \node[mirror=(A)--(B), transform shape, circle, inner sep=1.5pt, minimum width=2pt, fill, color=red, above=0 of a0,label=-270:$\textcolor{red}{a_1}$] (a1) {};
      \node[mirror=(B)--(C), transform shape, circle, inner sep=1.5pt, minimum width=2pt, fill, color=red, above=0 of a0,label=-270:$\textcolor{red}{a_2}$] (a2) {};
      \node[mirror=(C)--(D), transform shape, circle, inner sep=1.5pt, minimum width=2pt, fill, color=red, above=0 of a0,label=0:$\textcolor{red}{a_3}$] (a3) {};
      \node[mirror=(D)--(E), transform shape, circle, inner sep=1.5pt, minimum width=2pt, fill, color=red, above=0 of a0,label=0:$\textcolor{red}{a_4}$] (a4) {};
      \node[mirror=(E)--(A), transform shape, circle, inner sep=1.5pt, minimum width=2pt, fill, color=red, above=0 of a0,label=0:$\textcolor{red}{a_5}$] (a5) {};

      \draw[thick,-,dash pattern=on 2pt off 2pt] (a0) -- (a1);
      \draw[thick,-,dash pattern=on 2pt off 2pt] (a0) -- (a2);
      \draw[thick,-,dash pattern=on 2pt off 2pt] (a0) -- (a3);
      \draw[thick,-,dash pattern=on 2pt off 2pt] (a0) -- (a4);
      \draw[thick,-,dash pattern=on 2pt off 2pt] (a0) -- (a5);

      \filldraw[blue] (a0) circle (2pt) node[above right] {$a_0$};
\end{tikzpicture}
}
    \caption{A convex polytope defined by the intersection of half-spaces $Ax \leq b$ and its NN representation by the ``reflection'' argument. The interior of the polytope is closest to $a_0$ and the exterior is closest to one of $a_1,\dots,a_5$.}
    \label{fig:polygon}
\end{figure}
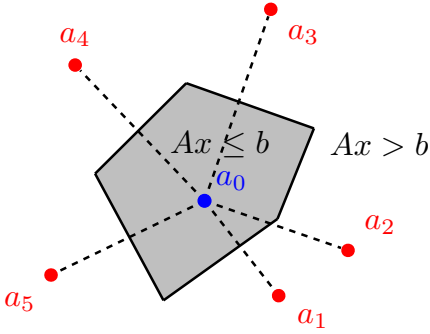

\begin{theorem}
    \label{th:polytope}
    Let $A \in \mathbb{Z}^{m \times n}$ and $b \in \mathbb{Z}^m$ define a convex polytope in $\mathbb{R}^n$ by the intersection of half-spaces $AX \leq b$. Then, there exists an NN representation with $m+1$ anchors and resolution $O(RES(\diag(AA^T))$.
\end{theorem}
\begin{proof}
    Without loss of generality, we assume that there is an binary interior point in the convex polytope $AX \leq b$. To ensure strict feasibility, we can modify $AX \leq b$ to $AX \leq b + 0.5$ where both inequalities have the same solution sets for $X \in \{0,1\}^n$ and if there is no such binary vector, then $\mathds{1}\{AX \leq b\}$ is in fact a constant Boolean function with value $0$.

    The geometrical idea is to compute the reflection of the feasible point $a_0 \in \{0,1\}^n$ with respect to $Ax = b$. We take $a_0$ a positive anchor and we consider the other anchors negative as follows:
    \begin{align}
        a_i = a_0 + 2c_iA_i
    \end{align}
    for $i \in \{1,\dots,m\}$ where $A_i$ denotes the $i^{th}$ row of the matrix $A$ and $c_i$ is a real constant that will be specified later. We also assume that $a_0 + c_iA_i$ is a point on the $A_i^T x = b_i$ to correctly reflect $a_0$ to $a_i$. Then, we have $A_i^T(a_0 + c_iA_i) = A_i^T a_0 + c_i||A_i||_2^2 = b_i$ and therefore $c_i = \frac{b_i - A_i^T a_0}{||A_i||_2^2}$. Implicitly, all $c_i > 0$ because of the feasibility condition.

    Note that whenever $A_i^T X = b_i$, the anchors $a_i$ and $a_0$ are equidistant and we use a classical perturbation argument: $AX \leq b$ has the same solution set as $AX \leq b + 0.5$ for $X \in \{0,1\}^n$.

    When we expand the squared Euclidean form $d(a_i,X)^2$, we obtain
    \begin{align}
        \nonumber
        d(a_i,X)^2 &= |X| - 2a_0^T X + ||a_0||_2^2 \\ 
        &\hphantom{aaaaaaa}- 4c_i(A_i^T X - A_i^T a_0) +4c_i^2 ||A_i||_2^2 \\ \nonumber
        &= d(a_0,X)^2 + \frac{4}{||A_i||_2^2} (b_i - A_i^T a_0)(A_i^T a_0 - A_i^T X) \\
        &\hphantom{aaaaaaa}+ \frac{4}{||A_i||_2^2} (b_i - A_i^T a_0)^2 \\ 
        &= d(a_0,X)^2 + \frac{4}{||A_i||_2^2} (b_i - A_i^T a_0)(b_i - A_i^T X)
    \end{align}
    We have two cases: either $A_i^T X < b_i$ for all $i \in \{1,\dots,m\}$ or $A_i^T X > b$ for some $i \in \{1,\dots,m\}$. To compare the distance of $X$ to $a_0$ and $a_i$, we simplify the expression to
    \begin{align}
        \label{eq:distance_a_i}
        d(a_i,X)^2 - d(a_0,X)^2 = \frac{4}{||A_i||_2^2} (b_i - A_i^T a_0)(b_i - A_i^T X)
    \end{align}
    
    \textbf{\underline{Case 1:}} If $A_i^T X < b_i$ for all $i \in \{1,\dots,m\}$, we need $d(a_i,X)^2 - d(a_0,X)^2 > 0$. Then, since $b_i - A_i^T a_0 > 0$ by definition, the RHS of Eq. \eqref{eq:distance_a_i} is greater than $0$ if and only if $A_i^T X < b_i$ for all $i \in \{1,\dots,m\}$.

    \textbf{\underline{Case 2:}} If $A_i^T X > b_i$ for some $i \in \{1,\dots,m\}$, we need $d(a_i,X)^2 - d(a_0,X)^2 < 0$ for such $i$. Then, since $b_i - A_i^T a_0 > 0$ by definition, the RHS of Eq. \eqref{eq:distance_a_i} is less than $0$ if and only if $A_i^T X > b_i$.

    For $A_i^T X > b_i$, we do not care which anchor is closest to $x$ among $a_i$ for $i \in \{1,\dots,m\}$ since they all have the same labeling. Therefore, the proposed scheme is indeed an NN representation for the convex polytope defined by $AX \leq b$.

    We now prove the resolution bound. For $c_i = \frac{b_i - A_i^T a_0}{||A_i||_2^2}$, we see that $b_i - A_i^T a_0 \leq 2||A_i||_2^2$ loosely. We assume that the bias term can be at most $||A_i||_2^2$ otherwise the threshold function is trivial. Then, the resolution of the $a_0 + 2c_iA_i$ is $O(RES(\diag(AA^T))$ and the claim holds by considering all Euclidean norms.
\end{proof}



We can also add equality conditions to $AX \leq b$ in that the AND of two linear threshold functions $\mathds{1}\{A_i^T X \leq b_i\}$ and $\mathds{1}\{A_i^T X \geq b_i\}$ implies $\mathds{1}\{A_i^T X = b_i\}$. We can perturb the bias terms by $0.5$ and obtain $\mathds{1}\{A_i^T X = b_i\}$ as the AND of $\mathds{1}\{A_i^T X \leq b_i + 0.5\}$ and $\mathds{1}\{A_i^T X \geq b_i - 0.5\}$ so that there is always an interior point $a_0 \in \{0,1\}^n$ in the polytope. Otherwise, there is no binary interior point and the function is the constant function $f(X) = 0$.

\begin{corollary}
    \label{cor:polytope_eq}
    For a system of linear equations $Ax = b$ where $A \in \mathbb{Z}^{m\times n}$ and $b \in \mathbb{Z}^m$, there exist an NN representation with $2m+1$ anchors and resolution $O(RES(\diag(AA^T))$ checking if $AX = b$ or not for $X \in \{0,1\}^n$.
\end{corollary}



\section{NN Representations of Depth-2 Circuits with Symmetric Top Gate}
\label{sec:sym_lt_elt}

For $\text{AND} \circ \text{LT}$ or $\text{AND} \circ \text{ELT}$, what happens if we replace the AND with the OR gate? For $\text{OR} \circ \text{LT}$, the answer is easy because NN representations are closed under complement operation (as we can revert the labels of anchors) and the complement of $\text{OR} \circ \text{LT}$ is $\text{AND} \circ \text{LT}$, therefore, we already have a solution by Theorem \ref{th:polytope}. However, for $\text{OR} \circ \text{ELT}$, we cannot do the same as the complement of an exact threshold function needs not to be exact.

Obtaining a construction for $\text{OR} \circ \text{ELT}$ is not straight-forward and the arbitrary case is still unresolved. We define the following set of \textit{regularity conditions}. Let $W \in \mathbb{Z}^{m\times n}$ be the weight matrix of the first layer and $b \in \mathbb{Z}^m$ be the bias vector of the first layer.

\begin{enumerate}
    \item The weights of each gate has the same norm. $||W_i||_2^2 = ||W_j||_2^2$ for all $i,j$.
    \item The weights of each gate are mutually orthogonal, $W_i^T W_j = 0$ for all $i \neq j$.
    \item There exists an $X^* \in \{0,1\}^n$ such that $WX^* = b$.
\end{enumerate}

The regularity conditions hurt the generality but the result is still very useful and applicable to many functions. For example, if all gates have \textit{disjoint} inputs and have the same norm for their weights, then all conditions are satisfied.

\begin{theorem}
    \label{th:sym_elt}
    Suppose that there is an $n$-input Boolean function $f(X)$ such that $f(X) \in \textup{SYM}\circ\textup{ELT}$ obeying the regularity conditions with $m$ many gates in the first layer. Let $g(Z)$ be the top symmetric Boolean function where $Z \in \{0,1\}^m$.
    
    There exists an NN representation of $f(X)$ with $\sum_{t \in \mathcal{T}} \binom{m}{m-t}2^{m-t}$ many anchors where $\mathcal{T} = \{I_0+1,I_1+1,\dots,I_{I(g)-1}+1\}$ contains the left interval boundaries for the top symmetric gate $g(Z)$. The resolution of the construction is $O(\log{m} + RES(\diag(WW^T)))$.
\end{theorem}
\begin{proof}
    The anchors we construct are as follows.
    \begin{align}
        a_{jk}^{(t)} = X^* + du_{jk}^{(t)}
    \end{align}
    We will design $u_{jk}^{(t)}$ and $d \in \mathbb{Z}$ will be picked later. $t \in \mathbb{Z}$ is the \textit{type} of the anchor and $j,k$ are indices to be determined later. We denote $w_i$ to be the $i$\textsuperscript{th} row of the weight matrix $W$.

    The squared Euclidean distances will be
    \begin{align}
        \label{eq:squared_dist_u}
        \nonumber
        d(a_{jk}^{(t)},X)^2 &= |X| - 2 X^T X^* + ||X^*||_2^2 \\
        &\hphantom{aaa}-2du_{jk}^{(t)} (X - X^*) + d^2 ||u_{jk}^{(t)}||_2^2
    \end{align}

    Since $|X| - 2X^T X^* + ||X^*||_2^2$ is the same for all anchors, we do not care its value when we compare distances. Now, we pick $u_{jk}^{(t)}$ as the all plus-minus combinations of the $m-t$ combinations of all $m$ weight vectors in $W$ for any $t \in \{0,\dots,m\}$. That is, there are $\binom{m}{m-t}$ selections of weight vectors for each $t$ for $u_{jk}^{(t)}$. We also use an abuse of notation here: $\pm w_1 \pm w_2$ is a compact form to denote all $\{w_1 + w_2, w_1 - w_2, -w_1 + w_2, -w_1 - w_2\}$. We can write all $u_{jk}^{(t)}$s as follows.

    \begin{align}
        \hphantom{aa}&u_{jk}^{(0)} \in \Big\{\pm w_1 \pm w_2 \pm \dots \pm w_{m}\Big\} \\
         \hphantom{aa}&u_{jk}^{(1)} \in \Big\{\pm w_1 \pm \dots \pm w_{m-1},\dots, \pm w_2 \pm \dots \pm w_m \Big\} \\ \nonumber
        &\vdots \\
        \hphantom{aa}&u_{jk}^{(m-1)}  \in \Big\{\pm w_1, \pm w_2, \dots, \pm w_{m}\Big\} \\
        \hphantom{aa}&u_{jk}^{(m)}  \in \emptyset
    \end{align}

    Now, we define $u_{jk}^{(t)}$s more precisely. Let $j_1,\dots,j_{m-t}$ be the binary expansion of $(j-1) \in \mathbb{Z}$ with $\{0,1\}$ entries. The index $j$ denotes the unique sign pattern for $w$s in $u_{jk}^{(t)}$.
    
    For the anchor type $t$, we define the family of index sets $\mathcal{F}^{(t)} = \binom{[m]}{m-t}$. Alternatively, $\mathcal{F}^{(t)} = (\mathcal{I}_k^{(t)})_{k \in K^{(t)}}$. Here $\mathcal{I}_k^{(t)}$ is an index set with size $\binom{m}{m-t}$ from the elements $\{1,\dots,m\}$ and $K^{(t)} = \{1,\dots,\binom{m}{m-t}\}$. In other words, for $u_{jk}^{(t)}$s, the index $k$ denotes a specific $m-t$ selection out of $\{1,\dots,m\}$. We thus obtain
    \begin{align}
        \label{eq:u_def}
        u_{jk}^{(t)} &= \sum_{i = 1}^{m-t} (-1)^{j_i} w_{{\mathcal{I}_k^{(t)}}_i} \text{ for } t < m \\
        u_{jk}^{(m)} &= 0
    \end{align}
    For instance, consider $m = 7$, $t = 4$, $(j - 1) = 2$, and $k=6$. First, we enumerate all $\binom{7}{3}$ combinations in an order to find the correct selection of weights.
    \begin{align}
    \nonumber
        &\Big\{\{w_1,w_2,w_3\}, \{w_1,w_2,w_4\}, \{w_1,w_2,w_5\},\{w_1,w_2,w_6\}, \\
        &\hphantom{\Big\{}\{w_1,w_2,w_7\},\{w_1,w_3,w_4\}, \dots, \{w_5,w_6,w_7\} \Big\}
    \end{align}
    The corresponding index sets are the subscripts of the weights, e.g., $\mathcal{I}_4^{(4)} = \{1,2,6\}$.
    For $k = 6$, we pick $\mathcal{I}_6^{(4)} = \{1,3,4\}$. Since $j-1$ has the binary expansion for $(j_1,j_2,j_3) = (0,1,0)$, we obtain the sign pattern $(1,-1,1)$, implying that we have $u_{36}^{(4)} =w_1-w_3+w_4$.

    Let $g(Z)$ be the symmetric Boolean function for the top gate of the circuit with $I(g)$ many intervals where $Z \in \{0,1\}^m$. In the construction, we only include $u_{jk}^{(t)}$ for $t \in \mathcal{T} = \{I_0+1,I_1+1,\dots,I_{I(g)-1}+1\}$.

    \textbf{\underline{Claim:}} If $Z$ is in the $l^{th}$ interval of $g(Z)$, i.e., $|Z| \in [I_{l-1}+1,I_l]$, then the closest anchor is of type $t = I_{l-1}+1$. That is,
    \begin{align}
        \nonumber
        \label{eq:type_min}
        I_{l-1} + 1 = &\arg\min_{t} -2d{u_{jk}^{(t)}}^T(X-X^*) + d^2||u_{jk}^{(t)}||_2^2 \text{ where } \\
        &w_i^T X = b_i \text{ for } |Z| \text{ many $i$s from } \{0,1,\dots,m\}
    \end{align}

    If the claim is true, then depending on the value $g(Z)$ for that interval, we label all the corresponding type of anchors to \textcolor{red}{\textbf{0}} or \textcolor{blue}{1} and the construction follows.

    To prove the claim, we first pick $d = \frac{1}{m||w||_2^2}$ where $||w||_2^2 = ||w_1||_2^2 = \dots = ||w_m||_2^2$ and by orthogonality, we have $||u_{jk}^{(t)}||_2^2 = (m-t)||w||_2^2$. Then, from Eq. \eqref{eq:squared_dist_u}, we obtain 
    \begin{align}
    \nonumber
        &-2d{u_{jk}^{(t)}}^T(X-X^*) + d^2||u_{jk}^{(t)}||_2^2 \\
        &\hphantom{aaaa}=\frac{-2}{m||w||_2^2}\Big({u_{jk}^{(t)}}^T (X - X^*) - 0.5\frac{(m-t)}{m}\Big) 
    \end{align}

    Suppose that $u_{jk}^{(t)}(X - X^*) > u_{jk}^{(v)} (X - X^*)$ for any $v \in \{0,\dots,m\}$. Then, recalling that $X^* \in \{0,1\}^n$,
    \begin{align}
        &\Big({u_{jk}^{(t)}}^T (X - X^*) - 0.5\frac{(m-t)}{m}\Big) \\
        &\hphantom{aaaa}\geq  \Big({u_{jk}^{(v)}}^T (X - X^*) + 1 - 0.5\frac{(m-t)}{m}\Big) \\
        &\hphantom{aaaaaa} \geq \Big({u_{jk}^{(v)}}^T (X - X^*) + 0.5 \Big) \\
         &\hphantom{aaaaaaa} > \Big({u_{jk}^{(v)}}^T (X - X^*) \Big) \\
         &\hphantom{aaaaaaaaa} > \Big({u_{jk}^{(v)}}^T (X - X^*) - 0.5\frac{(m-v)}{m}\Big)
    \end{align}
    and thus, the $t$ value that minimizes Eq. \eqref{eq:type_min} always have the largest $u_{jk}^{(t)}(X - X^*)$ value.

    If $u_{jk}^{(t)}(X - X^*) = u_{jk}^{(v)} (X - X^*)$ for some $v$, then 
    \begin{align}
        &\Big({u_{jk}^{(t)}}^T (X - X^*) - 0.5\frac{(m-t)}{m}\Big) \\
        &\hphantom{aaaa} > \Big({u_{jk}^{(v)}}^T (X - X^*) - 0.5\frac{(m-v)}{m}\Big)
    \end{align}
    only if $t > v$. Therefore, the $t$ value that minimizes Eq. \eqref{eq:type_min} is the maximum among the $t$ values maximizing $u_{jk}^{(t)}(X - X^*)$.
    
    Expanding $u_{jk}^{(t)}(X-X^*)$ using Eq. \eqref{eq:u_def}, we obtain
    \begin{equation}
        u_{jk}^{(t)}(X - X^*) = \sum_{i = 1}^{m-t} (-1)^{j_i} (w_{{\mathcal{I}_k^{(t)}}_i}^T X - b_{{\mathcal{I}_k^{(t)}}_i})
    \end{equation}
    Given the value of $|Z|$ where $z_i = \mathds{1}\{w_i^T X = b_i\}$ for $i \in \{1,\dots,m\}$, let $\mathcal{I}$ be the index set where $w_i^T X \neq b_i$. It is actually a selection of $m - |Z|$ out of $m$ values and therefore, $\mathcal{I} = \mathcal{I}_k^{(|Z|)}$ for a combination enumerated by some $k$. It can be seen that $u_{jk}^{(t)}(X - X^*)$ is maximized if $\mathcal{I}_k^{(|Z|)} = \mathcal{I}$ and whether $w_i^T X < b_i$ or $w_i^T X > b_i$ is of no importance since there is a pair $j,k$ such that
    \begin{equation}
        \label{eq:max_u}
        \max_{j,k,t} {u_{jk}^{(t)}}^T (X - X^*) = \sum_{i \in \mathcal{I}} |w_i^T X - b_i|
    \end{equation}

    The optimal $t$ value is less than or equal to $|Z|$ because $\mathcal{I}$ contains $m-|Z|$ values. Any superset for the index set $\mathcal{I}$ will include $w_i^T X = b_i$ for some $i$ and the value of the dot product $u_{j'k'}^T(X - X^*)$ cannot increase. Mathematically, for the optimal $j,k,j',k'$ choices,
    \begin{align}
        \nonumber
        &{u_{jk}^{(p)}}^T (X - X^*) = {u_{j'k'}^{(q)}}^T (X-X^*) \\
        &\hphantom{aaaaaa}\text{ for } p,q \in \{0,\dots,I_{l-1}+1\} \\
        \nonumber
        &{u_{jk}^{(p)}}^T (X - X^*) > {u_{j'k'}^{(q)}}^T (X-X^*) \\
        \nonumber
        &\hphantom{aaaaaa}\text{ for } p \in \{0,\dots,I_{l-1}+1\} \text{ and } \\
        &\hphantom{aaaaaa fori} q \in \{I_l + 1,\dots,I_{I(g)-1}+1\}
    \end{align}
    
    By our previous observation about the maximal $t$, we conclude that $t = I_{l-1}+1$. This proves the claim and hence, completes the validation of the construction. The number of anchors is $\sum_{t \in \mathcal{T}} \binom{m}{m-t}2^{m-t}$, which is easy to verify.

    For the resolution, we see that $\frac{1}{m||w||_2^2} (w_1 + \dots + w_m)$. The maximum value for the numerator can be $m||w||_2^2$ and it holds that $RES(A) = O(\log{m} + RES(\diag(WW^T)))$.
\end{proof}

Theorem \ref{th:sym_elt} is powerful since it provides constructions to many important functions in Circuit Complexity Theory. Hence, it is an important milestone to find NN representations with size and resolution upper bounds.

\begin{corollary}
    \label{cor:sym_eq}
    Let $f(X)$ be an $2mn$-input Boolean function $\textup{SYM}_m \circ \textup{EQ}_{2n}$ where there are $m$ many disjoint $n$-input $\textup{EQ}$ functions in the first layer. Then, we obtain the following table of results.
\begin{table}[H]
\floatsetup[table]{style=plaintop}
  \captionsetup[table]{position=top}
\label{tab:sym_eq}
\centering
\begin{tabular}{|c||c||c|}
    \hline
     \multirow{2}{*}{\textbf{Function}} & \multicolumn{2}{c|}{\textbf{NN Representation}} \\
     \cline{2-3} & Size & Resolution \\
     \hline\hline
     $\text{AND}_m \circ \text{EQ}_{2n}$ & $2m + 1$ & $O(n)$ \\
     \hline
     $\text{OR}_m \circ \text{EQ}_{2n}$ & $(m+2)2^{m-1}$ & $O(n)$ \\
     \hline
     $\text{PARITY}_m \circ \text{EQ}_{2n}$ & $3^m$ & $O(n)$ \\
     \hline
\end{tabular}
\end{table}

\end{corollary}

First of all, the $\text{AND}_m \circ \text{EQ}_{2n}$ result also follows from Corollary \ref{cor:polytope_eq} and it holds that Theorem \ref{th:sym_elt} is a generalization of it in the case of same norm weights. For $\text{OR} \circ \text{EQ}$, we see that $\mathcal{T} = \{0,1\}$ and for $\text{PARITY} \circ \text{EQ}$, $\mathcal{T} = \{0,1,2,\dots,m\}$ where the representation size formula corresponds to the binomial expansion of $(1+2)^m$. Hence, changing the top gate could greatly change the number of anchors in the NN representation.

We also obtain a construction for $\text{SYM} \circ \text{LT}$. Compared to the Theorem \ref{th:sym_elt}, the difference is in the set of intervals boundaries where Theorem \ref{th:sym_lt} includes both boundaries. In contrast, because each linear threshold function requires $2$ anchors instead of $3$ as for exact threshold functions, we do not see an extra exponential term $2^{m-t}$ in the binomial sum. 

\begin{theorem}
    \label{th:sym_lt}
    Suppose that there is an $n$-input Boolean function $f(X)$ such that $f(X) \in \textup{SYM}\circ\textup{LT}$ obeying the regularity conditions with $m$ many gates in the first layer. Let $g(Z)$ be the top symmetric Boolean function where $Z \in \{0,1\}^m$.
    
    There exists an NN representation of $f(X)$ with $\sum_{t \in \mathcal{T}} \binom{m}{m-t}$ many anchors where $\mathcal{T} = \{I_1,I_1+1,\dots,I_{I(g)-1}, I_{I(g)-1}+1\}$ contains the all the interval boundaries for the top symmetric gate $g(Z)$. The resolution of the construction is $O(\log{m} + RES(\diag(WW^T)))$.
\end{theorem}
\begin{proof}
    The proof is essentially the same as the proof of Theorem \ref{th:sym_elt} except we do not have $j$ parameter in the anchor construction. All sign patterns are fixed to $-1$s.
    Then, we have
    \begin{align}
        a_{k}^{(t)} &= X^* + du_{k}^{(t)} \\
        u_{k}^{(t)} &= \sum_{i=1}^{m-t} -w_{{\mathcal{I}_k^{(t)}}_i} \text{ for t < m} \\
        u_{k}^{(m)} &= 0
    \end{align}
    for the same definition of the family of index sets. We again pick $d = \frac{1}{m||w||_2^2}$. 
    
    We also assume that the bias vector $b$ is replaced with $b-0.5$ to ensure the equality cases do not appear for the linear threshold functions in the first layer. In this case, we need to be careful about the choice of $X^*$ since there is no $WX = b - 0.5$ such that $X \in \{0,1\}^n$. We remedy this by using the first two regularity conditions so that $X^*$ in the third regularity condition is replaced by $X^* - \frac{1}{2||w||_2^2}(w_1 + \dots + w_m)$ where $||w||_2^2$ is the norm of a row of $W$. This does not increase the resolution asymptotically at the end since $RES(X^*) = RES(\diag(WW^T)$.
    
    Given the value of $|Z|$ where $z_i = \mathds{1}\{w_i^T X > b_i\}$ for $i \in \{1,\dots,m\}$ and $|Z| \in [I_{l-1}+1,I_{l}]$, let $\mathcal{I}$ be the index set $w_i^T X < b_i$. It is actually a selection of $m - |Z|$ out of $m$ values and $u_k^{(t)}(X-X^*)$ is maximized if $\mathcal{I}_k^{(|Z|)} = \mathcal{I}$.
    \begin{equation}
        \max_{k,t} {u_{k}^{(t)}}^T (X - X^*) = \sum_{i \in \mathcal{I}} (b_i - w_i^T X)
    \end{equation}
    For the optimal $k,k'$ choices, we have
    \begin{align}
        \nonumber
        &{u_{k}^{(p)}}^T (X - X^*) < {u_{k'}^{(p+1)}}^T (X-X^*) \\
        \label{eq:lower_part}
        &\hphantom{aaaaaa}\text{ for } p \in \{0,\dots,|Z|-1\} \\
        \nonumber
        &{u_{k}^{(p)}}^T (X - X^*) > {u_{k'}^{(p+1)}}^T (X-X^*) \\
        \label{eq:upper_part}
        &\hphantom{aaaaaa}\text{ for } p \in \{|Z|,\dots,m\}
    \end{align}
    Since $|Z| \in [I_{l-1}+1,I_l]$, the optimal $t$ value will either be $I_{l-1}+1$ or $I_l$. Since we include both interval boundaries in this construction and both types $t = I_{l-1}+1$ and $t = I_l$ have the same label, the construction follows. The resolution is similar to the previous proof and the increase in the resolution of $X^*$ makes no difference asymptotically.
\end{proof}


\begin{corollary}
    \label{cor:sym_lt}
    Let $f(X)$ be an $mn$-input Boolean function $\textup{SYM}_m \circ \textup{LT}_{n}$ where there are $m$ many disjoint $n$-input $\textup{LT}$ functions in the first layer with same norm weight vectors. Then, we obtain the following table of results.
\begin{table}[H]
\floatsetup[table]{style=plaintop}
  \captionsetup[table]{position=top}
\label{tab:sym_lt}
\centering
\begin{tabular}{|c||c||c|}
    \hline
     \multirow{2}{*}{\textbf{Function}} & \multicolumn{2}{c|}{\textbf{NN Representation}} \\
     \cline{2-3} & Size & Resolution \\
     \hline\hline
     $\text{AND}_m \circ \text{LT}_{n}$ & $m+1$ & $O(n\log{n})$ \\
     \hline
     $\text{OR}_m \circ \text{LT}_{n}$ & $m+1$ & $O(n\log{n})$ \\
     \hline
     $\text{PARITY}_m \circ \text{LT}_{n}$ & $2^m$ & $O(n\log{n})$ \\
     \hline
\end{tabular}
\end{table}
\end{corollary}

It is remarkable that for the $\text{IP2}_{2n}$ function, Corollary \ref{cor:sym_lt} provides a construction with $2^n$ anchors and constant resolution since $\text{IP2}_{2n} = \text{PARITY}_n \circ \text{AND}_2 \in \text{SYM} \circ \text{LT}$. This is far from the lower bound where there is no explicit resolution bound. For constant resolution, the construction we provide could be optimal. We also see that  $\text{AND}_m \circ \text{LT}_{n}$ and $\text{OR}_m \circ \text{LT}_{n}$ have the same size representations.

\section{NN Representations of LDLs and EDLs}
\label{sec:ldl_edl}

The main characteristic of decision lists is the \textit{domination principle} where the threshold gates in the higher location will determine the output independent of what the lower level threshold gate outputs are. We design the location of the anchors based on this observation. The geometric approach to find an NN representation for LDLs is shown in Figure \ref{fig:ldl_construction}. 

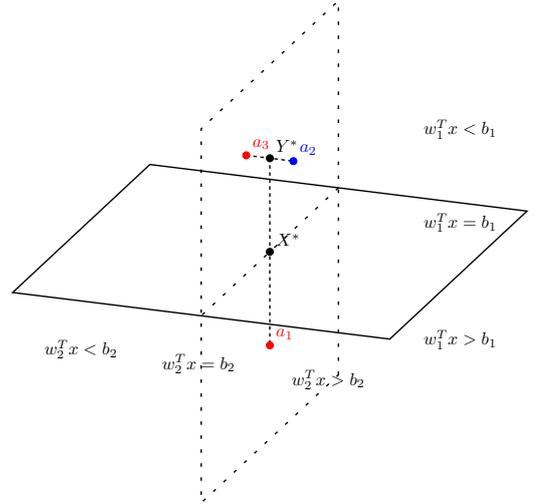
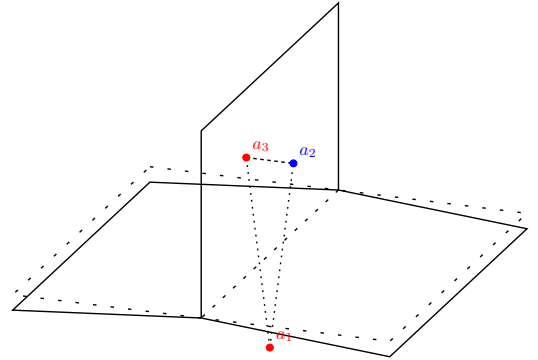
\begin{figure}[h!]
    \centering
    \begin{subfigure}{\textwidth}
    \centering
    \scalebox{0.66}{
        \tdplotsetmaincoords{70}{110}
        \begin{tikzpicture}[tdplot_main_coords]
            \newcommand{\eq}{=}
            \draw[thick,-,dash pattern=on 2pt off 6pt] (-4,0,-4) -- (-4,0,4) -- (4,0,4) -- (4,0,-4) -- cycle;
            \draw[thick,-] (-4,-4,0) -- (-4,4,0) -- (4,4,0) -- (4,-4,0) -- cycle;
            

            \draw[thick,-,dash pattern=on 2pt off 6pt](-4,0,0)--(4,0,0);
            
            \coordinate (a0) at (0,0,0);

            \node[label=0:$w_1^T x < b_1$] (w1) at ($(a0) + (0,3,3)$) {};
            \node[label=0:$w_1^T x \eq b_1$] (w2) at ($(a0) + (0,3,1)$) {};
            \node[label=0:$w_1^T x > b_1$] (w3) at ($(a0) + (0,3,-1.5)$) {};

            \node[label=0:$w_2^T x > b_2$] (w4) at ($(a0) + (7,2.75,0)$) {};
            \node[label=0:$w_2^T x \eq b_2$] (w5) at ($(a0) + (7,0,0)$) {};
            \node[label=0:$w_2^T x < b_2$] (w6) at ($(a0) + (7,-2.5,0)$) {};
            
            \coordinate (X) at ($(a0)$);

            \filldraw[black] (X) circle (2pt) node[above right]  {$X^*$};

            
            \coordinate (a21) at ($(a0) + (0,0,2)$);
            
            \coordinate (a22) at ($(a0) + (0,0,-2)$);
            \draw[thick,-,dash pattern=on 2pt off 2pt](a21)--(a22);

            \coordinate (a11) at ($(a21) + (0,0.5,0)$);
            
            \coordinate (a12) at ($(a21) + (0,-0.5,0)$);
            
            \draw[thick,-,dash pattern=on 2pt off 2pt](a12)--(a11);
            
            \filldraw[blue] (a11) circle (2pt) node[above right]  {$a_2$};
            
            \filldraw[red] (a12) circle (2pt) node[above right]  {$a_3$};
            
            \filldraw[black] (a21) circle (2pt) node[above right]  {$Y^*$};
            
            \filldraw[red] (a22) circle (2pt) node[above right]  {$a_1$};
        \end{tikzpicture}
    }
    \caption{Anchor placement idea for the NN Representation for an LDL of depth $2$. Each anchor takes care of a leaf of the LDL.}
    \label{sub:ldl_1}
    \end{subfigure}
    
    \begin{subfigure}{\textwidth}
    \centering
    \scalebox{0.66}{
        \tdplotsetmaincoords{70}{110}
        \begin{tikzpicture}[tdplot_main_coords]
            \newcommand{\eq}{=}
            \draw[thick,-] (-4,0,0) -- (-4,0,4) -- (4,0,4) -- (4,0,0);
            \draw[thick,-,dash pattern=on 2pt off 8 pt] (-4,-4,0) -- (-4,4,0) -- (4,4,0) -- (4,-4,0) -- cycle;
            
            \draw[thick,-] (4,0,0) -- (4,-4,-0.33) -- (-4,-4,-0.33) -- (-4,0,0);

            \draw[thick,-] (-4,0,0) -- (-4,4,-0.33) -- (4,4,-0.33) -- (4,0,0);

            \draw[thick,-,dash pattern=on 2pt off 6 pt](-4,0,0)--(4,0,0);
            
            \coordinate (a0) at (0,0,0);
            
            \coordinate (X) at ($(a0)$);


            \coordinate (a21) at ($(a0) + (0,0,2)$);
            
            \coordinate (a22) at ($(a0) + (0,0,-2)$);

            \coordinate (a11) at ($(a21) + (0,0.5,0)$);
            
            \coordinate (a12) at ($(a21) + (0,-0.5,0)$);
            
            \draw[thick,-,dash pattern=on 2pt off 2pt](a12)--(a11);

            \draw[thick,-,dash pattern=on 1pt off 3pt](a22)--(a11);

            \draw[thick,-,dash pattern=on 1pt off 3pt](a22)--(a12);
            
            \filldraw[blue] (a11) circle (2pt) node[above right]  {$a_2$};
            
            \filldraw[red] (a12) circle (2pt) node[above right]  {$a_3$};
            
            \filldraw[red] (a22) circle (2pt) node[above right]  {$a_1$};
        \end{tikzpicture}
    }
    \caption{The approximate decision regions for the NN representation. The closer $a_2$ and $a_3$ are to each other, the better the bottom region will approximate a half-space.}
    \label{sub:ldl_2}
    \end{subfigure}
    
    \caption{Anchor placement idea for the NN Representation for an LDL of depth $2$. In this example, the labels of the anchors are arbitrary.}
    \label{fig:ldl_construction}
\end{figure}

\begin{theorem}
    \label{th:ldl}
    Suppose that an $n$-input Linear Decision List $l(X)$ of depth $m$ is given under regularity conditions with a weight matrix $W \in \mathbb{R}^{m\times n}$. Then, there is an NN representation for $l(X)$ with $m+1$ anchors and resolution $O(mRES(\diag(WW^T)))$.
\end{theorem}
\begin{proof}
    For linear threshold gates, the third regularity condition contradicts the desirable property that there should be no $X \in \{0,1\}^n$ on the hyperplanes themselves. We will describe a way to solve this issue by modifying this condition as it is done in the proof of Theorem \ref{th:sym_lt}.
    
    We previously perturbed the bias terms by $0.5$ to ensure that there is no binary vectors on the hyperplanes. Alternatively, this can be done by seeing that $\mathds{1}\{w^T X \geq b\}$ is equivalent to $\mathds{1}\{2w^T X \geq 2b - 1\}$. Although both methods are equivalent, this change makes the steps in the proof cleaner.
    
    Given that there is a binary vector $X^* \in \{0,1\}^n$ such that $WX^* = b$, we can find $X' = X^* - \frac{1}{2||w||_2^2}(w_1+ \dots + w_m)$ so that $2WX' = 2b - 1$ where $||w||_2^2$ is the norm of the original weights. In this case, the resolution of the weights increases by $1$ bit because of the doubling and the resolution of $X^*$ is at most $O(RES(||w||_2^2)$.
    
    Therefore, without loss of generality, we can modify the third regularity condition in the following way: There exists $X^* \in \mathbb{Q}^n$ such that $WX^* = b$ with ``sufficiently'' small resolution and there is no binary $X \in \{0,1\}^n$ such that $WX = b$. 
    
    To imitate the linear decision list and obtain the domination principle, we construct the anchors as follows where $c_i = \Big(\frac{1}{2||w||_2^2}\Big)^i$.
    
    \begin{align}
        a_i &= X^* - \sum_{j < i} c_jw_j + c_iw_i \text{ for } i = 1,\dots,m \\
        a_{m+1} &= X^* - \sum_{j < m} c_jw_j - c_mw_m
    \end{align}
    The labeling of $a_i$ directly corresponds to the labels of $z_i$ for the decision list. We claim that for any $k$, if the location of the leading one in the decision list is $k$, i.e., $(w_1^T X < b_1,\dots,w_{k-1}^T X < b_{k-1}, w_k^T X > b_k, \times,\dots,\times)$ with $\times$ being \textit{don't cares}, then $a_k$ is the closest to the input vector. Hence, the following two conditions are necessary and sufficient. Roughly, the first condition states that if the output is $1$, $a_k$ dominates all the rest and the second condition states that if the output is $0$, one should proceed to the next anchor.
    
    \begin{align}
        \label{eq:ldl_cond_1}
        &w_k^T X > b_k \Rightarrow d(a_k,X)^2 - d(a_l,X)^2 < 0 \hphantom{a}\forall k < l \\
        \nonumber
        \label{eq:ldl_cond_2}
        &w_k^T X < b_k \Rightarrow d(a_k,X)^2 - d(a_{k+1},X)^2 > 0 \hphantom{a} \\
        &\hphantom{aaaaaaaaaaaaaaaaaaaaa}\forall k \in \{1,\dots,m\}
    \end{align}
    
    Using the orthogonality of $w_i$s, the squared Euclidean distances can be written as
    \begin{align}
        \nonumber
        d(a_i,X)^2 &= |X| - 2X^TX^* + ||X^*||_2^2 \\
        \nonumber
        &\hphantom{aa}+2\sum_{j < i}c_j(w_j^T X - b_j) + ||w||_2^2\sum_{j<i} c_j^2 \\
        \label{eq:ldl_distance}
        &\hphantom{aa}- 2c_i(w_i^T X - b_i) + c_i^2||w||_2^2
    \end{align}
    For the condition in Eq. \eqref{eq:ldl_cond_1}, we obtain the following.
    \begin{align}
        \nonumber
        &d(a_k,X)^2 - d(a_l,X)^2 = -4c_k(w_k^T X - b_k) \\
        \nonumber
        &\hphantom{aaa}- 2\sum_{j=k+1}^{l-1} c_j(w_j^T X - b_j) \\
        \label{eq:ldl_sq_difference}
        &\hphantom{aaa}+ 2c_l(w_l^T X - b_l) - ||w||_2^2 \sum_{j=k+1}^l c_j^2 < 0
    \end{align}
    We optimize this inequality in an adversarial sense where the contribution of the negative terms are smallest and of the positive terms are largest possible. Note that we bound $|w_j^T (X - X^*)| \leq 2||w||_2^2$ loosely. We see that $w_k^T X - b_k = 1$ and $w_j^T X - b_j = -2||w||_2^2$ for $j > k$ gives the tightest bounds. $j = m$ or $j=m+1$ does not matter. Then, putting the values of $c$s, we get
    \begin{align}
        \nonumber
        &d(a_k,X)^2 - d(a_l,X)^2 = -\frac{4}{2^k||w||_2^{2k}} \\
        \nonumber
        &\hphantom{aaaaaaaa}+\sum_{j=k+1}^{l} \frac{2}{(2||w||_2^2)^{j-1}} \\
        \label{eq:ldl_k_l}
        &\hphantom{aaaaaaaa}-\sum_{j=k+1}^l \frac{1}{2^{2j} ||w||_2^{4j+2}} < 0
    \end{align}
    The first term dominates the second term for any finite value of $l$ using a geometric series argument. For $l \rightarrow \infty$,
    \begin{align}
        &-\frac{4}{(2||w||_2^2)^k} + \frac{2}{(2||w||_2^2)^k} \frac{1}{1-\frac{1}{2||w||_2^2}} \\
        &\hphantom{aaaa}= -\frac{4}{(2||w||_2^2)^k} + \frac{2}{(2||w||_2^2)^k} \frac{2||w||_2^2}{2||w||_2^2-1} \leq 0
    \end{align}
    The fraction $\frac{2||w||_2^2}{2||w||_2^2-1}$ is at most $2$ for $||w||_2^2 = 1$ and the expression is strictly negative for $||w||_2^2 > 1$. Due to the third negative term in Eq. \eqref{eq:ldl_k_l}, the claim is true.

    The proof for the second condition (Eq. \eqref{eq:ldl_cond_2}) is similar. We first compute $d(a_k,X)^2 - d(a_{k+1},X)^2$ and consider $k < m$.
    \begin{align}
        \nonumber
        &d(a_k,X)^2 - d(a_{k+1},X)^2 = -4c_k(w_k^T X - b_k) \\
        &\hphantom{aaaaaaaa}+2c_{k+1}(w_{k+1}^T X  - b_{k+1}) - c_{k+1}^2 ||w||_2^2 > 0 
    \end{align}
    Since $w_k^T X - b_k < 0$, we take the value $-1$, making the contribution of the first term positive and small. Similarly, we take $w_{k+1}^T X - b_{k+1} = -2||w||_2^2$. Hence,
    \begin{align}
        &\frac{4}{(2||w||_2^2)^k} - \frac{2}{(2||w||_2^2)^{k}} - \frac{1}{(2^{2k+2}||w||_2^{4k+2})} \\
        \label{eq:ldl_bound_last}
        &\hphantom{aaa}=\frac{2}{(2^k||w||_2^{2k})} - \frac{1}{(2^{2k+2}||w||_2^{4k+2})} > 0
    \end{align}
    The last inequality holds since $||w||_2^2 \geq 1$ and $k \geq 1$.
    
    Finally, we consider $k = m$ separately and since $w_m^T X < b_m$, we have
    \begin{align}
        \nonumber
        &d(a_m,X)^2 - d(a_{m+1},X)^2 = -4c_m(w_m^T X - b_m) > 0
    \end{align}
    The resolution claim follows by how small $c_m = \frac{1}{(2||w||_2^2)^m}$ is. Therefore, it becomes $O(m RES(\diag(WW^T)))$. The resolution of the point $X^*$ is bounded by the resolution of $||w||_2^2$ so asymptotically, it makes no difference in the resolution of the anchors.
\end{proof}

In addition, we can replace the regularity conditions in Theorem \ref{th:ldl} only with $m \leq n$ where $m$ is the depth of the list. Let $A^+$ denote the Moore-Penrose inverse of a matrix $A \in \mathbb{R}^{m\times n}$.

\begin{theorem}
    \label{th:ldl_generalized}
    Suppose that an $n$-input Linear Decision List $l(X)$ of depth $m$ is given with a weight matrix $W \in \mathbb{Z}^{m\times n}$ where $m \leq n$ and a bias vector $b \in \mathbb{Z}^m$. Then, there is an NN representation for $l(X)$ with $m+1$ anchors and resolution $O(RES(W^+)+mRES(\diag(WW^T)))$.
\end{theorem}
\begin{proof}
    First of all, this proof depends on the proof of Theorem \ref{th:ldl} with an explicit algorithm to find $X^*$, which is the most crucial step and will be done later in the proof. 
    
    We first assume that $w_i^T X \neq b_i$ for any $X \in \{0,1\}^n$ without loss of generality simply by changing the bias vector $b$ to $b - 0.5$. Compared to Theorem \ref{th:ldl}, since we do not assume any regularity conditions, this change is sufficient. In addition, we can assume without loss of generality that $W$ is full-rank by another perturbation argument. For example, for very small $\epsilon > 0$, we take $W' = W + \epsilon I_{m,n}$ where $I_{m,n}$ is the $m \times n$ sub-identity matrix with the first $m$ rows of the $I_{n\times n}$ and the threshold functions $\mathds{1}\{w_i^T X \geq b_i - 0.5\}$ are the same as $\mathds{1}\{w_i^T X + \epsilon X_i \geq b_i - 0.5\}$ when $\epsilon < 1/2$. Let $W_{m\times m}$ denote the $m \times m$ sub-matrix of $W$. $W'$ is full-rank if and only if
    \begin{align}
    \nonumber
        \det(W'(W')^T) &= \det\Big(WW^T + \epsilon(W_{m\times m}+(W_{m\times m})^T) \\
        \label{eq:determinant}
        &\hphantom{aa}+ \epsilon^2 I_{m\times m}\Big) \neq 0
    \end{align}
    Because Eq. \eqref{eq:determinant} is a finite polynomial in $\epsilon$ with at most $2m$ many roots, there are infinitely many choices for $0 < \epsilon < 1/2$.
    
    We have the same construction in the proof of Theorem \ref{th:ldl} with an additional assumption: Whenever we subtract $c_iw_i$ from $X^* - \sum_{j<i}c_jw_j$, the vector should satisfy $w_{i+1}^T x = b_{i+1}$. Define $X_{(i)}^* = X^* - \sum_{j<i}c_jw_j$ and $X_{(1)}^* = X^*$.
    \begin{align}
        a_i &= X^* - \sum_{j < i} c_jw_j + c_iw_i \text{ for } i = 1,\dots,m \\
        a_{m+1} &= X^* - \sum_{j < m} c_jw_j - c_mw_m
    \end{align}
    where $w_i^T X_{(i)}^* = b_i$.

    Under this assumption, we see that the squared distance differences are equivalent to the ones in the proof of Theorem \ref{th:ldl} (see Eq. \eqref{eq:ldl_sq_difference}). For $i < m$, we have
    \begin{align}
        \nonumber
        d(a_i,X)^2 &= |X| - 2X^TX_{(i)}^* + ||X_{(i)}^*||_2^2 \\
        \nonumber
        &\hphantom{aa}- 2c_i(w_i^T X - w_i^T X_{(i)}^*) + c_i^2||w_i||_2^2 \\
        \nonumber
        &= |X| - 2X^T(X^*-\sum_{j<i}c_jw_j) + ||X^*-\sum_{j<i}c_jw_j||_2^2 \\
        \nonumber
        &\hphantom{aa}- 2c_i(w_i^T X - b_i) + c_i^2||w_i||_2^2\\
        \nonumber
        &= |X| - 2X^TX^* + ||X^*||_2^2 \\
        \nonumber
        &\hphantom{aa}+2\sum_{j < i}c_j(w_j^T X - w_j^T X^*) + \Big(\sum_{j<i} c_jw_j\Big)^2 \\
        \nonumber
        &\hphantom{aa}- 2c_i(w_i^T X - b_i) + c_i^2||w_i||_2^2\\
        \nonumber
        &= |X| - 2X^TX^* + ||X^*||_2^2 \\
        \nonumber
        &\hphantom{aa}+2\sum_{j < i}\Big(c_j(w_j^T X - w_j^T X^* + \sum_{k < j} c_kw_j^Tw_k)\Big)\\
        \nonumber
        &\hphantom{aa}+ \sum_{j<i} c_j^2||w_j||_2^2- 2c_i(w_i^T X - b_i) + c_i^2||w_i||_2^2\\
        \nonumber
         &= |X| - 2X^TX^* + ||X^*||_2^2 \\
        \nonumber
        &\hphantom{aa}+2\sum_{j < i}c_j(w_j^T X - w_j^T (X^* - \sum_{k < j} c_kw_k))\\
        \nonumber
        &\hphantom{aa}+ \sum_{j<i} c_j^2||w_j||_2^2- 2c_i(w_i^T X - b_i) + c_i^2||w_i||_2^2\\
        \nonumber
        \label{eq:ldl_distance_new}
        &= |X| - 2X^TX^* + ||X^*||_2^2 \\
        \nonumber
        &\hphantom{aa}+2\sum_{j < i}c_j(w_j^T X - b_j) + \sum_{j<i} c_j^2||w_j||_2^2 \\
        &\hphantom{aa}- 2c_i(w_i^T X - b_i) + c_i^2||w_i||_2^2
    \end{align}
    One can observe that Eq.\eqref{eq:ldl_distance_new} is the same as Eq. \eqref{eq:ldl_distance} except the squared norms of $w_i$. Therefore, for a correct selection of $c_i$s, the construction should follow by satisfying the conditions Eq. \eqref{eq:ldl_cond_1} and \eqref{eq:ldl_cond_2}. We pick $c_i = c_{i-1} 2||w||_2^2$ and $c_1 = 1$ for simplicity where $||w||_2^2 = \max_k ||w_k||_2^2$. By similar bounding arguments, the same steps in the proof of Theorem \ref{th:ldl} could be followed.
    
    Now, we have to find an explicit $X^*$ to conclude the proof. Recall that $w_i^T X_{(i)}^* = b_i$ where $X_{(i)}^* = X^* - \sum_{j < i}c_jw_j$. Then, $w_i^T X^* = b_i + \sum_{j<i}c_jw_i^T w_j$. Clearly, this defines a linear system $WX^* = B$ where $B \in \mathbb{Q}^{m}$ and $B_i = b_i + \sum_{j<i}c_jw_i^T w_j$. Since $W$ is full-rank without loss of generality, there always exists $X^* = W^+B$ which solves the system exactly. 

    We observe that $b_i \leq ||w_i||^2$ and $|w_i^T w_j| \leq \max_k ||w_k||_2^2$. Since $c_i$s shrink geometrically, the resolution is $O(m\diag(WW^T))$ so that $RES(B) = mRES(\diag(WW^T))$. Hence,
    the resolution of $X^*$ becomes $O(RES(W^+) + mRES(\diag(WW^T)$. We conclude that the resolution of the construction is $O(RES(W^+) + mRES(\diag(WW^T)))$.
\end{proof}

\begin{corollary}
    \label{cor:sym_ldl}
    For any symmetric Boolean function $f(X)$ with $I(f)$ many intervals, $NN(f) \leq I(f)$. 
\end{corollary}

While this complexity result on symmetric Boolean functions is striking, it has already been proven \cite{kilic2023information}. Corollary \ref{cor:sym_ldl} follows by the simple construction of symmetric Boolean functions as LDLs. For instance, an LDL construction for an $8$-input symmetric Boolean function $f(X)$ with $I(f) = 5$ (see Eq. \eqref{eq:counterexample}) is given in Figure \ref{fig:ldl_sym} where each interval has the form $[I_{i-1}+1,I_{i}]$ for $i \in \{1,\dots,5\}$ (take $I_0 = - 1$).

\begin{equation}
    \label{eq:counterexample}
    \begin{tabular}{c|c c c}
        $|X|$ & $f(X)$ \\
        \cline{1-2}
         0 & \textcolor{blue}{1} & \hphantom{a} & $I_{1} = 0$\\
         1 & \textcolor{red}{\textbf{0}} & \hphantom{a} & $I_{2} = 1$\\ 
         2 & \textcolor{blue}{1} & \hphantom{a} & \\
         3 & \textcolor{blue}{1} & \hphantom{a} & \\
         4 & \textcolor{blue}{1} & \hphantom{a} & \\
         5 & \textcolor{blue}{1} & \hphantom{a} & \\
         6 & \textcolor{blue}{1} & \hphantom{a} & $I_{3} = 6$\\
         7 & \textcolor{red}{\textbf{0}} & \hphantom{a} & $I_4 = 7$\\
         8 & \textcolor{blue}{1} & \hphantom{a} & $I_5 = 8$
    \end{tabular}
\end{equation}

\begin{figure}[h!]
    \centering
    \begin{tikzpicture}
    \newdimen\nodeDist
    \nodeDist=20mm
    \node [rectangle,draw] (A) {\footnotesize$\mathds{1}\{|X| \leq 0\}$};
    \path (A) ++(-120:\nodeDist) node [rectangle,draw] (B) {\footnotesize$\mathds{1}\{|X| \leq 1\}$};
    \path (A) ++(-60:\nodeDist) node (C) {\textcolor{blue}{1}};
    
    \path (B) ++(-120:\nodeDist) node [rectangle,draw] (D) {\footnotesize$\mathds{1}\{|X| \leq 6\}$};
    \path (B) ++(-60:\nodeDist) node (E) {\textcolor{red}{\textbf{0}}};
    
    \path (D) ++(-120:\nodeDist) node [rectangle,draw] (F) {\footnotesize$\mathds{1}\{|X| \leq 7\}$};
    \path (D) ++(-60:\nodeDist) node (G) {\textcolor{blue}{1}};

    \path (F) ++(-120:\nodeDist) node (H)
    {\textcolor{blue}{1}};
    \path (F) ++(-60:\nodeDist) node (I) {\textcolor{red}{\textbf{0}}};
    
    \draw (A) -- (B) node [left,pos=0.25] {0};
    \draw (A) -- (C) node [right,pos=0.25] {1};
    \draw (B) -- (D) node [left,pos=0.25] {0};
    \draw (B) -- (E) node [right,pos=0.25] {1};
    \draw (D) -- (F) node [left,pos=0.25] {0};
    \draw (D) -- (G) node [right,pos=0.25] {1};
    \draw (F) -- (H) node [left,pos=0.25] {0};
    \draw (F) -- (I) node [right,pos=0.25] {1};
\end{tikzpicture}
    \caption{A Linear Decision List of depth $4$ for the symmetric Boolean function in Eq.\eqref{eq:counterexample} with $I(f) = 5$.}
    \label{fig:ldl_sym}
\end{figure}
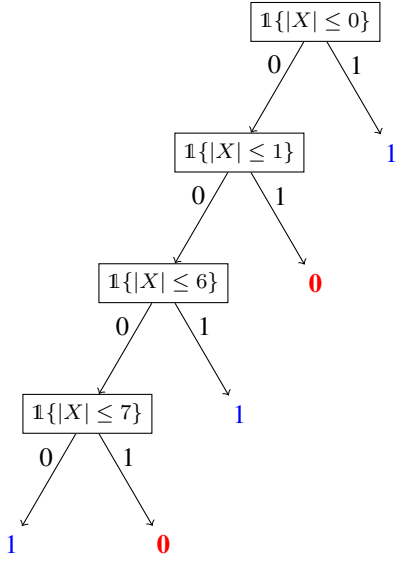

\begin{conjecture}
    Let $f(X)$ be a Boolean function. Then, $NN(f) \leq LDL(f) + 1$ where $LDL(f)$ is the smallest depth of linear decision lists computing $f(X)$.
\end{conjecture}

By Theorem \ref{th:ldl_generalized}, the conjecture is true when $LDL(f) \leq n$.

We now focus on the NN representations of EDLs. For EDLs, we give a construction idea similar to circuits of $\text{OR}\circ\text{ELT}$ with a slight change in parameters to implement the domination principle. 

\begin{theorem}
    \label{th:edl}
    Suppose that an $n$-input Exact Decision List $l(X)$ of depth $m$ is given under regularity conditions. Then, there is an NN representation for $l(X)$ with $(m+1)2^m$ anchors and resolution $O(\log{m} + RES(\diag(WW^T)))$.
\end{theorem}
\begin{proof}
    We essentially construct a NN representation for the $\text{OR} \circ \text{ELT}$ type of circuits (consider Corollary \ref{cor:sym_eq}) and modify it to obtain the \textit{domination principle}. We consider the anchors as follows similar to the proof of Theorem \ref{th:sym_elt} with two types. We assume $d > c_i$ and all $c_i > 0$ except $d = c_{m+1} > 0$.
    \begin{align}
        a_{jk} &= X^* + d u_{jk} + (-1)^{j_m} c_k w_k \text{ for } k \in \{1,\dots,m\} \\
        a_{j(m+1)} &= X^* + d u_{jm} + (-1)^{j_m} c_{m+1} w_m
    \end{align}
    where $u_{jk} = \pm w_1 \pm \dots \pm w_{k-1} \pm w_{k+1} \pm \dots \pm w_m$ (only $w_k$ is excluded) for $k \in \{1,\dots,m\}$ and $j_m \in \{0,1\}$. Also, $c_{m+1} = d$. The sign pattern is given by the binary expansion of $j-1$ in $m-1$ bits as in the proof of Theorem \ref{th:sym_elt}. For example, for $m = 5$, $j-1 = 4$ gives $(j_1,j_2,j_3,j_4) = (0,0,1,0)$ and $u_{52} = w_1+w_3-w_4+w_5$. If, in addition, $j_5 = 0$, then we find $a_{52} = d(w_1 + w_3 -w_4 + w_5) + c_2w_2$. In comparison, if $j_5 = 1$, we obtain $a_{(20)2} = d(w_1 + w_3 - w_4 + w_5) - c_2w_2$.

    We have the following squared Euclidean norm expression for this construction.
    \begin{align}
    \nonumber
        d(a_{jk},X)^2 &= |X| - 2X^T X^* + ||X^*||^2 \\
        \nonumber
            &\hphantom{aaa} -2du_{jk}^T(X - X^*) + d^2||u_{jk}||_2^2 \\
            \nonumber
            &\hphantom{aaa} -2j_m c_k (w_k^T X - b_k) + c_k^2 ||w_k||_2^2 \\
            &\hphantom{aaa} -2j_mdc_k w_k^T u_{jk}
    \end{align}

    By the orthogonality assumption and the constructions of $u_{jk}$s, we have $w_k^T u_{jk} = 0$. Since our goal is to find the anchor minimizing the Euclidean distance, there is a $j = j^*$ such that
    \begin{align}
    \nonumber
        d(a_{j^*k},X)^2 &= |X| - 2X^T X^* + ||X^*||^2 \\
        \nonumber
            &\hphantom{aaa} -2d\sum_{i \neq k} |w_i^TX - b_i| + d^2(m-1)||w||_2^2 \\
            &\hphantom{aaa} -2c_k |w_k^T X - b_k| + c_k^2 ||w||_2^2
    \end{align}
    Given the optimal selection of $j$s minimizing the Euclidean distance, we have to find the argument $k$ which will globally minimize this.
    
    For an EDL, we see that $z_l$ is picked if and only if $(w_1^T X \neq b_1,\dots, w_{l-1}^T X \neq b_{l-1}, w_l^T X = b_l, \times,\dots, \times)$ where we have inequalities for $k < l$ and a \textit{don't care} region for $l < k$. We will deal with $(w_1^T X \neq b_1, \dots, w_m^T X \neq b_m)$ later.
    
    We want $a_{j^*k}$ to be the closest anchor to $X$ for $k = l$ and some $j^*$ for $(w_1^T X \neq b_1,\dots, w_{l-1}^T X \neq b_{l-1}, w_l^T X = b_l, \times,\dots, \times)$. Hence, when we compare different $k,l$ we get
    \begin{align}
        \label{eq:edl_nec}
    \nonumber
        &d(a_{j^*l},X)^2 - d(a_{j^+k},X)^2 \\
        &\hphantom{aaaa}= -2(d-c_k)|w_k^T X - b_k| + (c_l^2 - c_k^2)||w||_2^2 < 0
    \end{align}
    Note that $w_l^T X = b_l$ so that term does not appear.
    
    \textbf{\underline{Case 1 $(l < k)$:}} This is the simple case. The inequality in Eq. \eqref{eq:edl_nec} is the tightest when $|w_k^T X - b_k| = 0$. Then, for $k \leq m$, we obtain $c_l < c_k$ for $l < k$ as a necessary condition. $k = m+1$ is trivial since $d > c_i$ for all $i \in \{1,\dots,m\}$ and $d(a_{j^*l},X)^2 - d(a_{j^+(m+1)},X)^2 = (c_l^2 - d^2)||w||_2^2 < 0$.

    \textbf{\underline{Case 2 $(k < l \leq m)$:}} The tightest Eq. \eqref{eq:edl_nec} becomes is when $|w_k^T X - b_k| = 1$ and we have
    \begin{equation}
        c_k^2 - \frac{2}{||w||_2^2} c_k + 2\frac{d}{||w||_2^2} - c_l^2 > 0
    \end{equation}
    Let $d = 1/||w||_2^2$ and $c_i = \frac{i}{(m+1)||w||_2^2}$ for $i \in \{1,\dots,m\}$. Then, we obtain
    \begin{equation}
        \frac{k^2}{(m+1)^2} - 2\frac{k}{m+1} + 2 - \frac{l^2}{(m+1)^2} > 0
    \end{equation}
    Since $l \leq m$, the tightest this inequality becomes is when the value of the fourth term is $1$. Then, we obtain
    \begin{align}
        \Big(\frac{k}{(m+1)^2} - 1\Big)^2 > 0
    \end{align}
    which is true for since $k \neq m+1$.

    \textbf{\underline{Case 3 $(l = m+1)$:}} Finally, we consider $(w_1^T X \neq b_1, \dots, w_m^T X \neq b_m)$. For this case, we claim that for any $k \in \{1,\dots,m\}$,
    \begin{align}
    \nonumber
        &d(a_{j^*(m+1)},X)^2 - d(a_{j^+k},X)^2 \\
        &\hphantom{aaaa}= -2(d-c_k)|w_k^T X - b_k| + (d^2 - c_k^2)||w||_2^2 < 0
    \end{align}
    Take $|w_k^T X - b_k| = 1$ similarly. Then,
    \begin{align}
        c_k^2 - \frac{2}{||w||_2^2} c_k + 2\frac{d}{||w||_2^2} - d^2 > 0
    \end{align}
    Consider $d = 1/||w||_2^2$ and $c_i = \frac{i}{(m+1)||w||_2^2}$ for $i \in \{1,\dots,m\}$. We get
    \begin{align}
     \frac{k^2}{(m+1)^2} -2\frac{k}{m+1} + 2 - 1  > 0 \\
    \Big(\frac{k}{m+1} - 1\Big)^2 > 0
    \end{align}
    which is true since $k \neq m+1$.
    
    This shows that the construction works. The size of representation is $(m+1)2^{m}$ by counting through $j$s and $k$s. Similar to the proof of Theorem \ref{th:sym_elt}, the resolution is $O(\log{m} + RES(\diag(WW^T)))$.
\end{proof}

We note that the idea for Theorem \ref{th:edl} works for LDLs as well with $(m+1)2^m$ many anchors and a possible resolution improvement from $O(m)$ to $O(\log{m})$.

\begin{corollary}
    \label{cor:edl}
    Let $f(X)$ be the $2nm$-input Boolean function $\textup{OMB}_m \circ \textup{EQ}_{2n}$ where there are $m$ many disjoint $2n$-input EQ functions in the first layer. Then, there is an NN representation with $(m+1)2^m$ anchors and $O(n)$ resolution.
\end{corollary}
\section{Conclusion and Future Directions}
\label{sec:conclusion}

In this work, the NN representations of $\text{SYM}\circ\text{LT}$, $\text{SYM}\circ\text{ELT}$, $\text{DOM}\circ\text{LT}$, and $\text{DOM}\circ\text{ELT}$ are treated. These circuits include many important functions in Boolean analysis and Circuit Complexity Theory. Novel NN constructions are provided for some of these functions. Finding similar constructions for $\text{LT}\circ\text{LT}$ and removing the regularity constraints are future challenges. In addition, the treatment of circuits with depth $3$ or higher could give new insights to the theory of NN representations.

\clearpage
\printbibliography

\end{document}